\documentclass[12pt]{article}
\usepackage{graphicx}
\usepackage{amsmath}
\usepackage{amssymb}
\usepackage{amsthm}
\usepackage{amsfonts}
\usepackage{xcolor}
\usepackage{thm-restate}
\usepackage[T1]{fontenc}
\usepackage{hyperref}
\newtheorem{red}{Reduction Rule}{\bfseries}{\itshape}
\newtheorem{bred}{Branching Rule}{\bfseries}{\itshape}
\newcommand{\claimqed}{\hfill $\square$} 

\usepackage{fullpage}

\newtheorem{theorem}{Theorem}
\newtheorem{lemma}{Lemma}

\newtheorem{claim}{Claim}[section]

\newtheorem{definition}{Definition}[section]

\title{MaxMin Separation Problems: FPT Algorithms for $st$-Separator and Odd Cycle Transversal}

\author{Ajinkya Gaikwad\thanks{Indian Institute of Science Education and Research, Pune, India \href{mailto:ajinkya.gaikwad@students.iiserpune.ac.in}{ajinkya.gaikwad@students.iiserpune.ac.in}.}
\and
Hitendra Kumar\thanks{
Indian Institute of Science Education and Research, Pune, India \href{mailto:hitendra.kumar@students.iiserpune.ac.in}{hitendra.kumar@students.iiserpune.ac.in}.}
\and
Soumen Maity\thanks{
Indian Institute of Science Education and Research, Pune, India \href{mailto:soumen@iiserpune.ac.in}{soumen@iiserpune.ac.in}.}
\and Saket Saurabh\thanks{Institute of Mathematical Sciences, Chennai, India \href{mailto:saket@imsc.res.in}{saket@imsc.res.in}.}
\and Roohani Sharma\thanks{University of Bergen, Norway \href{mailto:r.sharma@uib.no}{r.sharma@uib.no}.}
}
\date{}

\begin{document}
\maketitle          
\begin{abstract}
In this paper, we study the parameterized complexity of the {\sc MaxMin} versions of two fundamental separation problems: {\sc Maximum Minimal $st$-Separator} and  {\sc Maximum Minimal Odd Cycle Transversal (OCT)}, both parameterized by the solution size.
In the {\sc Maximum Minimal $st$-Separator} problem, given a graph $G$, two distinct vertices $s$ and $t$ and a positive integer $k$, 
the goal is to determine whether there exists a minimal $st$-separator in $G$ of size at least $k$. 
Similarly, the {\sc Maximum Minimal OCT} problem seeks to determine if there exists a minimal set of vertices whose deletion results in a bipartite graph, and whose size is at least $k$.
We demonstrate that both problems are
fixed-parameter tractable parameterized by $k$.
Our FPT algorithm for {\sc Maximum Minimal $st$-Separator} answers the open question by Hanaka, Bodlaender, van der Zanden \& Ono [TCS 2019]. 

One unique   insight from this work is the following. We use the meta-result of Lokshtanov, Ramanujan, Saurabh \& Zehavi [ICALP 2018] that enables us to reduce our problems to highly unbreakable graphs.
This is interesting, as an explicit use of the recursive understanding and randomized contractions framework of Chitnis, Cygan, Hajiaghayi, Pilipczuk \& Pilipczuk [SICOMP 2016] to reduce to the highly unbreakable graphs setting (which is the result that Lokshtanov et al. tries to abstract out in their meta-theorem) does not seem obvious because certain ``extension'' variants of our problems are W[1]-hard.
\end{abstract}

\section{Introduction}\label{sec:intro}
In this work we study fixed-parameter tractability of two fundamental MaxMin separation problems: \textsc{Maximum Minimum $st$-Separator (MaxMin $st$-Sep)} and \textsc{Maximum Minimum Odd Cycle Transversal (MaxMin OCT)}.
In the \textsc{MaxMin $st$-Sep} problem, the input is an undirected graph $G$, two distinct vertices $s,t$ of $G$ and a positive integer $k$.
The goal is to determine if there exists a subset of vertices $Z$,
of size {\em at least} $k$,
such that $Z$ is a {\em minimal} $st$-separator in $G$. 
That is, the deletion of $Z$ disconnects $s$ and $t$ and the deletion of any proper subset of $Z$ results in a graph where $s$ and $t$ are connected.
Similarly, in the \textsc{MaxMin OCT} problem, given $G,k$, the goal is to determine if there 
exists a set of vertices $Z$, of size {\em at least} $k$, 
such that $Z$ intersects all odd length cycles in $G$
and $Z$ is minimal, that is no proper subset of $Z$ intersects all odd length cycles in $G$.

In contrast to the classical polynomial-time solvable \textsc{$st$-Separator} problem, where the goal is to find an $st$-separator of size {\em at most} $k$, 
the \textsc{MaxMin $st$-Sep} problem is NP-hard~\cite{HANAKA2019294}. 
The \textsc{MaxMin OCT} problem can also be shown to be NP-hard by giving a reduction from the {\sc MaxMin $st$-Sep} problem (see Section~\ref{MMOCT is NP-hard}).

In this work we show that both \textsc{MaxMin $st$-Sep} and \textsc{MaxMin OCT} are fixed-parameter tractable with respect to the parameter $k$. 
The first result, in fact,
resolves an open problem which was explicitly posed by Hanaka, Bodlaender, van der Zanden \& Ono~\cite{HANAKA2019294}.

{\sc MaxMin} versions of several classical vertex/edge deletion minimization problems have been studied in the literature. 
The original motivation behind studying such versions is that the size of the solution of the {\sc MaxMin} versions reflects on the worst-case guarantees of a greedy heuristic. 
In addition to this, the {\sc MaxMin} problems have received a lot of attention also because of their deep combinatorial structure which makes them stubborn even towards basic algorithmic ideas. 
As we will highlight later, neither greedy, nor an exhaustive-search strategy like branching, works for the {\sc MaxMin} versions of even the ``simplest'' problems without significant effort.
The {\sc MaxMin} versions are even harder to approximate. 
For example, the classic {\sc Vertex Cover} and {\sc Feedback Vertex Set} problems admit $2$-approximation algorithms~\cite{bafna19992}, which are tight under the Unique Games Conjecture~\cite{DBLP:journals/jcss/KhotR08}.
But the {\sc MaxMin Vertex Cover} admits a $n^{1/2}$-approximation, which is tight unless $P=NP$~\cite{BORIA201562,BONNET2018171}, and the {\sc MaxMin Feedback Vertex Set} admits a tight $n^{2/3}$-approximation~\cite{DUBLOIS202226}.

 We note here that the study of MaxMin versions of classical deletion problems is not the only proposed way of understanding worst-case heursitics guarantees (though in this work we only focus on such versions).
Several variations of different problems have been defined whose core is similar.
This includes problems like \textsc{$b$-Coloring}, \textsc{Grundy Coloring} etc, which are analogs for the classic \textsc{Chromatic Number} problem. 

\vspace{9pt}

\noindent\textbf{Parameterized Complexity of \textsc{MaxMin} problems in literature.}
 Given the extensive study of the {\sc Vertex Cover} problem in parameterized complexity, the parameterized complexity {\sc MaxMin Vertex Cover} has naturally garnered significant attention \cite{doi:10.1137/16M109017X,BORIA201562,BONNET2018171,10.1007/s00453-022-00979-z}. 
 More recently, {\sc MaxMin Feedback Vertex Set} problem has also been explored in \cite{DUBLOIS202226,lampis_et_al:LIPIcs.MFCS.2023.62}, where several faster FPT algorithms are proposed. 
 The {\sc MaxMin Upper Dominating Set} problem has been studied in \cite{DBLP:journals/corr/AbouEishaHLMRZ16,10.1007/s00453-022-01036-5,BAZGAN20182,DUBLOIS2022271},
 and its edge variant, the {\sc MaxMin Upper Edge Dominating Set}, is addressed in \cite{MONNOT202146,gaikwad2022parameterizedcomplexityupperedge}. 
 Additionally, {\sc MaxMin} and {\sc MinMax} formulations have been investigated for a range of other problems, including cut and separation problems \cite{DBLP:conf/icalp/Lampis21,DBLP:journals/corr/abs-2007-04513}, 
knapsack problems \cite{10.1007/978-3-642-40164-0_18,FURINI2017438}, matching problems \cite{10.1007/978-3-031-25211-2_29}, and coloring problems \cite{DBLP:journals/corr/abs-2008-07425}.
We elaborate more on the literature around {\sc MaxMin} versions of cut and separation
problems in the later paragraph.

We remark that for most problems mentioned in the paragraph above, 
showing that they are FPT parameterized by the solution size is not very difficult (though designing faster FPT algorithms could be much challenging and require deep problem-specific combinatorial insights).
The reason is that it is easy to bound the treewidth of the input graph: find a greedy packing of obstructions (edges for {\sc MaxMin Vertex Cover} and cycles for {\sc MaxMin Feedback Vertex Set}); if the packing size is at least $k$, then the instance is a yes-instance, otherwise there exists a vertex cover (resp. feedback vertex set) of the input graph, of size at most $2k$ (resp.~$\mathcal{O}(k \log k)$ from the Erd\"os-P\'osa theorem). In both case, the treewidth of the graph is bounded by a function of $k$. Since these problems are expressible in Monadic Second Order (MSO) Logic, from Courcelle's theorem a linear-time FPT in $k$ algorithm follows.

\vspace{9pt}

\noindent{\bf Parameterized Complexity of {\sc MaxMin} cut and separation problems in literature.}
The key challenge in the study of cut and separation problems, in contrast to the vertex/edge-deletion problems mentioned in the paragraph above, is that it is not always easy to bound the treewidth of the instances. 
Having said that the existing work on {\sc MaxMin} versions of cut and separation problems are still based on treewidth win-win approaches as explained below.

Hanaka et al.~\cite{HANAKA2019294} studied the parameterized complexity of the {\sc MaxMin Separator} problem. Here the input is only a connected graph $G$ and a positive integer $k$, and the goal is to find a minimal vertex set of size at least $k$ whose deletion disconnects the graph. This problem has an easy FPT algorithm parameterized by $k$ based on a win-win approach: if the input graph has large treewidth, then it has a large grid-minor which implies the existence of a large solution; otherwise the treewidth is bounded and since the problem is expressible in MSO, one can solve the problem in linear time on bounded treewidth graphs.
Note that this approach completely fails when we are looking for an $st$-separator (which is the problem we attack in the present work) for fixed vertices $s$ and $t$. In particular, a large grid minor does not necessarily imply a large $st$-separator.

In the same work Hanaka et al. designed an explicit FPT algorithm parameterized by treewidth for the {\sc MaxMin $st$-Sep} problem and left open the question of determining the parameterized complexity of {\sc MaxMin $st$-Sep} parameterized by the solution size $k$.

Common techniques like flow augmentation  \cite{10.1145/3641105} and treewidth reduction \cite{10.1145/2500119} are not suitable for solving the {\sc Maximum Minimal $st$-Separator} problem. This is because the size of the solution in this problem can be arbitrarily large, rather than being bounded by a function of $k$. When the solution size is unbounded, these methods fail to simplify the problem or give meaningful reductions.
We also considered the idea to reduce the problem to a scenario where the solution size is bounded by some function of $k$ using some win-win approach, and then use techniques like treewidth-reduction, unbreakable tree decompositions or even flow-augmentation.
But showing the first part seemed tricky and even if one shows that the solution size is bounded, it is still not straightforward to use these classical cut-based techniques as maintaining certificates of minimality of a partial solution is hard in general.

Later Duarte et al.~\cite{DBLP:journals/corr/abs-2007-04513} studied the edge-deletion variant of the {\sc MaxMin $st$-Sep} problem.
They termed it the {\sc Largest $st$-Bond} problem and showed, amongst other results, that it is FPT parameterized by the solution size $k$. At the core of this algorithm is another treewidth-based win-win approach, which seems hard to generalize to the vertex-deletion case.

\vspace{9pt}

\noindent{\bf Problem 1: MaxMin $st$-Sep.}
In this work, for the first time, we use the power of highly unbreakable instances to show fixed-parameter tractability of some {\sc MaxMin} separation problems.
For two positive integers $q,k$, a graph $G$ is called $(q,k)$-unbreakable if {\em no} vertex set of size {\em at most} $k$ can disconnect two (large) sets of size at least $q$ each. For the purposes of an informal discussion, we say a graph is unbreakable if it is $(q,k)$-unbreakable for some values of $q$ and $k$.

Chitnis et al.~\cite{doi:10.1137/15M1032077} developed a win-win approach based on this unbreakable structure of the graph. The approach has two parts: 
recursive understanding and randomized contractions. 
In the first part, a large enough part of the input graph is detected that is unbreakable and has small boundary to the rest of the graph. 
In the second part, a family of ``partial solutions'' are computed for this unbreakable part of the graph depending on how the solution for the whole graph interacts with its boundary. 
If the unbreakable part is large enough, then there exists an irrelevant edge/vertex that does not participate in the computed partial solutions, which helps in reducing the size of the graph.

Unfortunately, at the first glance it looks impossible to use this approach for solving any of the two problems {\sc MaxMin $st$-Sep} and {\sc MaxMin OCT}.
The problem is that in the above approach one needs to find partial solutions on unbreakable graphs for a more general ``extension-kind'' of problem.
In particular, the family of partial solutions should be such that if there exists a solution for the whole graph,
then one should be able to replace the part of this solution that intersects with the unbreakable part, with one of the computed partial solutions.
To do so, for example, it seems necessary to, in some implicit way at least, guess how an optimum solution of the whole graph intersects with the boundary of this unbreakable part and the partial solution should at least try to be ``compatible'' with this guess.
The bottleneck here is that given a vertex $v$, finding whether the graph $G$ has {\em any} minimal $st$-separator {\em containing $v$} is NP-hard (see Lemmas~\ref{irrelevant vertices} and \ref{induced path through a vertex}).
Thus, the problem of determining, given a subset of vertices $X$, whether $G$ has a minimal $st$-separator of size at least $k$, that contains $X$, is W[1]-hard (it is in fact para-NP-hard).

\begin{lemma}\label{irrelevant vertices}
Let $G$ be a graph containing vertices $s$ and $t$. A vertex $v \in V(G)$ is in some minimal
$st$-separator $Z$ if and only if there is an induced path between $s$ and $t$ containing $v$. 
\end{lemma}

\begin{proof} For the first part, assume there exists a minimal $st$-separator $Z$ that contains $v$. By the minimality of $Z$, there must be a path $P$ between $s$ and $t$ in the graph $G - (Z \setminus \{v\})$, which passes through $v$. In other words, path $P$ is such that no vertex in $Z$ other than $v$ appears on it. However, since there is no induced path between $s$ and $t$ through $v$, two vertices, say $a$ and $b$, on $P$ must be adjacent, with $a$ lying between $s$ and $v$, and $b$ lying between $v$ and $t$. This creates a new path between $s$ and $t$ that does not include any vertex from $Z$, contradicting the assumption that $Z$ is an $st$-separator.
\par For the second part, if there exists an induced path between $s$ and $t$ passing through $v$, we can construct a minimal $st$-separator $Z$ that includes $v$. According to Definition \ref{def:certificate of minimality}, let $S$ be the set of all vertices on the induced path from $s$ to $a$, and $T$ the set of all vertices on the induced path from $b$ to $t$, where $a$ is a predecessor of $v$ and $b$ a successor of $v$. These sets, $S$ and $T$, serve as a \textit{certificate} for the $st$-separator minimality for ${v}$. By Lemma \ref{greedy alg minimal st separator}, we can then construct a minimal $st$-separator that contains $v$. 
\end{proof}

\begin{lemma}[\cite{HAAS2006360}\label{induced path through a vertex}]
Given a graph $G$ and any three arbitrary vertices $s,t$ and $v$, determining if there exist an induced path between $s$ and $t$ through $v$ is NP-hard.
\end{lemma}

Therefore, a first glance suggests that computing ``partial solutions'' may be W[1]-hard in general. 
One can ask if the hardness holds even when the graph is unbreakable (which is our scenario).
It turns out yes, because a result by Lokshtanov et al.~\cite{Lokshtanov2018ReducingCM}, shows that the FPT algorithm for unbreakable graphs can be lifted to an FPT algorithm on general graphs for problems definable in Counting Monadic Second Order (CMSO) Logic. Since the extension version is also CMSO definable, there shouldn't be an FPT algorithm for the extension version even on unbreakable graphs.

Despite this issue, we show that the core lies in solving the problem on unbreakable graphs, and in fact the problem on unbreakable graphs is FPT using non-trivial insights.
In fact the result of Lokshtanov et al.~\cite{Lokshtanov2018ReducingCM} comes to rescue. In~\cite{Lokshtanov2018ReducingCM} Lokshtanov et al. show that, in order to show fixed-parameter tractability of CMSO definable problems parameterized by the solution size (say $k$), 
it is enough to design FPT algorithms for such problems when the input graph is $(q,k)$-unbreakable, for some large enough $q$ that depends only on $k$ and the problem. 
The highlight of this result, compared to explicitly doing the above-mentioned 
approach of recursive understanding and randomized contractions is that,
this results says that it is enough to solve the {\em same} problem on unbreakable graphs.
This allows us to skip taking the provably hard route of dealing with extension-kind problems.

After overcoming the above issues, we can safely say that the core lies in designing an FPT algorithm for the unbreakable case, which itself is far from obvious. We give a technical overview of this phase in Section~\ref{sec:tech-overview}.

\vspace{9pt}

\noindent{\bf Problem 2: MaxMin OCT.}
The challenges continue for the {\sc MaxMin OCT} problem.
In parameterized complexity, the first FPT algorithm for the classical {\sc Odd Cycle Transversal (OCT)} problem parameterized by the solution size, introduced the technique of Iterative Compression~\cite{DBLP:journals/orl/ReedSV04}. Using this, it was shown that  
{\sc OCT} reduces to solving at most $3^k \cdot n$ instances of the polynomial-time {\sc $st$-Separator} problem.
Most algorithms for different variants of the {\sc OCT} problem, like {\sc Independent OCT}, use the same framework and reduce to essentially solving the variant of {\sc $st$-Separator}, like {\sc Independent $st$-Separator}~\cite{DBLP:journals/talg/LokshtanovPSSZ20}.

For the {\sc MaxMin} variant, unfortunately this is not the case. 
In fact our FPT algorithm for {\sc MaxMin OCT} is independent of our FPT algorithm for {\sc MaxMin $st$-Sep}.
The reason is again that finding a minimal odd cycle transversal set (oct) that extends a given vertex is NP-hard (see Observation~\ref{delete unnecessary vertex} and Lemma~\ref{existence of odd cycle}). 

\begin{restatable}{observation}{obsirrelevant}\label{delete unnecessary vertex}
A vertex $v \in V(G)$ does not participate in any induced odd cycle of $G$ if and only if $v$ is not in any minimal OCT.
\end{restatable}

\begin{proof}For the forward direction,
let $Z$ be a minimal oct of $G$ of size at least $k$. 
Then
$G - Z$ is bipartite.
Additionally, for any $z \in Z$, 
the graph $G - (Z \setminus \{z\})$ contains an odd cycle through $z$, implying the existence of an induced odd cycle $C_z$ containing $z$.
As the cycle is an induced odd cycle, we know that the vertex $v$ does not lie on this cycle.
We will show that $Z$ is also a minimal oct in $G-v$.
Clearly, $G-(Z \cup \{v\})$ is a bipartite graph.
Therefore, $Z$ is an oct of $G-v$.
For any vertex $z\in Z$, 
there is an induced odd cycle $C_{z}$ in $G-(Z \cup \{v\})$.
Therefore, $Z$ is also a minimal odd cycle traversal of $G-v$.

For the other direction, let $Z$ be a minimal oct of $G - v$ with $|Z| \geq k$. 
Then, $G - (Z \cup \{v\})$ is bipartite.
For each $z \in Z$, the graph $G - ((Z \setminus \{z\}) \cup \{v\})$ contains an induced odd cycle $C_z$.
Since $v$ is not part of any induced odd cycle, $G - Z$ is bipartite, and $Z$ remains a minimal odd cycle transversal in $G$. 
\end{proof}

\begin{lemma}\label{existence of odd cycle}
  Given a graph $G$ and a vertex $v\in V(G)$, determining whether there is an induced odd cycle containing a given vertex is NP-complete.  
\end{lemma}
\begin{proof} We know that given two vertices $a$ and $b$, determining if a graph contains an induced path of odd length between $a$ and $b$ is NP-complete \cite{BIENSTOCK199185}.
Construct a graph $G'$ by adding a new vertex $x$ to $G$ and making it adjacent to vertices $a$ and $b$. It is clear that there exists an induced odd cycle through $x$ in $G'$ 
if and only if there is an induced path of odd length between $a$ and $b$ in $G$. This finishes the proof. 
\end{proof}

At the core of the Iterative Compression based approach for {\sc OCT},
a subset of vertices $X$ is guessed to be in the solution,
and after deleting $X$, the problem reduces to finding an $st$-separator, for some $s,t$.
In particular, the final solution is this set $X$ union a minimum $st$-separator.
For the {\sc MaxMin} case,
this amounts to finding a minimal $st$-separator that together with $X$ forms a minimal oct. Since the extension version of both the {\sc MaxMin $st$-Separator} and {\sc MaxMin OCT} are para-NP-hard, this leaves little hope to use the same approach for {\sc MaxMin} versions.

Having eliminated this approach, we again go back to the approach via unbreakable graphs. This time again the core lies in designing an FPT algorithm when the graph is highly unbreakable and this algorithm require lot more insights than that of the {\sc MaxMin $st$-Sep}. We give a technical overview of this stage in Section~\ref{sec:tech-overview}.

\subsection{Our results and technical overview}\label{sec:tech-overview}

Below we state two main theorems of this work.

\begin{restatable}{theorem}{thmstsep}\label{Mainthm}
  {\sc Maximum Minimal $st$-Separator} parameterized by $k$ is FPT.
\end{restatable}

\begin{restatable}{theorem}{thmoct}\label{thm:oct}
  {\sc Maximum Minimal OCT} parameterized by $k$ is FPT.
\end{restatable}

As discussed earlier, to prove both our theorems we first show that both {\sc MaxMin $st$-Sep} and {\sc MaxMin OCT} are CMSO definable. We then use Proposition~\ref{cmso}, to reduce to solving the problem on unbreakable graphs.

\begin{restatable}[Theorem 1, \cite{Lokshtanov2018ReducingCM}]{proposition}{prop}\label{cmso}
 Let $\psi$ be a CMSO formula. For all $c \in \mathbb{N}$, there exists $s \in \mathbb{N}$ such that if there
exists an algorithm that solves CMSO[$\psi$] 
on $(s, c)$-unbreakable structures in time $\mathcal{O}(n^{d})$ for some $d > 4$, then there exists an algorithm that solves CMSO[$\psi$] on general structures in time $\mathcal{O}(n^{d})$.
\end{restatable}

In particular, to prove Theorems~\ref{Mainthm} and~\ref{thm:oct}, 
it is enough to prove Theorems~\ref{thm:st-unbreak} and~\ref{thm:oct-unbreak}.

\begin{restatable}{theorem}{thmstsepunbreak}\label{thm:st-unbreak}
For positive integers $q,k \geq 1$,
{\sc Maximum Minimal $st$-Separator} on $(q,k)$-unbreakable graphs on $n$ vertices can be solved in time  $(k-1)^{2q} \cdot  n^{\mathcal{O}(1)}$.
\end{restatable}

\begin{restatable}{theorem}{thmoctunbreak}\label{thm:oct-unbreak}
For any positive integers \( q, k \geq 1 \), 
{\sc Maximum Minimal OCT} on \((q,k)\)-unbreakable graphs on $n$ vertices can be solved in time \( 2^{(qk)^{\mathcal{O}(q)}} \cdot n^{\mathcal{O}(1)} \).
\end{restatable}

\noindent{\bf \textsc{MaxMin $st$-Sep} on $(q,k)$-unbreakable graphs.} 
To prove Theorem~\ref{thm:st-unbreak}, 
we develop a branching algorithm which exploits the unbreakability of the input graph.
The key observation behind the algorithm is that for any two vertex sets $S$ and $T$ such that $s \in S$, $t \in T$ and $G[S]$ and $G[T]$ are connected, 
 every minimal $ST$-separator is also a minimal $st$-separator.
As we are working on $(q, k)$-unbreakable graphs, 
if we manage to find such sets where $|S| > q$ and $|T| > q$, 
then every minimal $ST$-separator has size at least $k$.
In this case, we can construct a minimal $ST$-separator greedily in polynomial time, which can then be returned as a solution. 
Therefore, the goal of our branching algorithm is to construct such sets $S$ and $T$. 

\sloppy The algorithm begins with $S = \{s\}$ and $T = \{t\}$.
We use a reduction rule which ensures that at any point $N(S)$ and $N(T)$ (the neighborhoods of $S$ and $T$) are minimal $ST$-separators in $G$.
Clearly, we can assume that $|N(S)| < k$ and $|N(T)| < k$; otherwise, we have already found a solution. 
A crucial observation is that if there exists a minimal $ST$-separator $Z$ of size at least $k$, then there exists a vertex $u \in N(S)\setminus N(T)$ (resp.~$u \in N(T) \setminus N(S)$) such that $u \notin Z$.
This implies that such a vertex $u$ remains reachable from $S$ even after removing the solution $Z$.
This observation allows us to grow the set $S$ to $S \cup \{u\}$ (resp.~$T$ to $T \cup \{u\}$).
Since $|N(S)| < k$ (resp.~$|N(T)| < k$),
one can branch on all possible vertices $u \in N(S)$ (resp.~$u \in N(T)$) and in each branch grow $S$ (resp.~$T$).
Based on whether $\min(|S|, |T|)$ is $|S|$ or $|T|$, we branch on the vertices in $N(S)\setminus N(T) $ or $N(T)\setminus N(S)$, respectively, to increase the size of these sets.
Once both $S$ and $T$ have sizes of at least $q$, we can greedily construct and output a minimal $ST$-separator (which is also a minimal $st$-separator) of size at least $k$.\\

\noindent{\bf \textsc{MaxMin OCT} on $(q,k)$-unbreakable graphs.}
The algorithm for this case uses two key lemmas, both of which provide sufficient conditions for a yes instance.
Here the input is a $(q,2k)$-unbreakable graph.

Our first key lemma (Lemma~\ref{long induced odd cycle}) says that 
 if there exists an induced odd cycle of length at least $2q+2$ in $G$, 
 then there always exist a minimal oct in $G$ of size at least $k$.
The proof of this result is based on a branching algorithm, which works similarly to the branching algorithm for the {\sc MaxMin $st$-Separator} problem on $(q, k)$-unbreakable graphs, by carefully selecting the sets $S$ and $T$ at the beginning of the algorithm.
Note that we cannot use this lemma, on its own as a a stopping criterion, because one does not know how to find a long induced odd length cycle efficiently. Nevertheless, as we see later it becomes useful together with our second sufficient condition.

Our second key lemma (Lemma~\ref{lot of small odd cycles}), which states a second sufficient condition,
says that if there exists a large enough family of distinct (and not necessarily disjoint) induced odd cycles in $G$ of lengths at most $2q+1$, 
then there always exists a minimal oct in $G$ of size at least $k$. 
The proof of this result relies on the Sunflower Lemma~\cite{erdos1960intersection,marekcygan}, along with the observation that the subgraph induced by the core of the sunflower must be bipartite. 
Such a large sunflower then serves as a certificate that any oct which is disjoint from the core of the sunflower is large. 
Also since the core is bipartite, there exists an oct that is disjoint from it.

Another simple, yet important observation is that, if a vertex is not part of any induced odd cycle, then deleting this vertex from the graph does not affect the solution. Again, the problem here is that determining whether a vertex passes through an induced odd cycle is NP-hard (see Lemma~\ref{existence of odd cycle}), 
therefore we cannot use it as a reduction rule.

Finally, using these two sufficient conditions and the observation above, 
we provide an FPT algorithm that, 
given a vertex $x$, either outputs an induced odd cycle containing $x$ if it exists, or concludes correctly that one of the two scenarios mentioned above occurs. In the later situation we are done. Also if the induced odd cycle containing $x$, which is returned, is long, then also we are done because of the first sufficient condition. If the algorithm outputs, that there is no induced odd cycle through $x$, then we can safely delete $x$ from the graph and reduce its size. Because of deleting such vertices, the new graph may no longer be unbreakable, in which case we start solving the problem on the reduced graph completely from scratch.

By running this algorithm for each $x \in V(G)$, we reduce the graph to one where each vertex is contained in some small induced odd cycle. 
If such a graph contains a sufficiently large number of vertices, 
it guarantees the existence of a large family of distinct small induced odd cycles in $G$, each of length at most $2q+2$, in which case we can return a yes instance because of the second sufficient condition. 
This result finally allows us to bound the number of vertices in the graph by $(qk)^{\mathcal{O}(q)}$.
We can then solve the problem  using a  brute-force algorithm on this graph.

\section{Preliminaries}\label{sec:prelims}
Throughout this article, $G=(V,E)$ denotes a finite, simple and undirected graph of order $|V|=n$.
The {\it (open) neighbourhood} $N_G(v)$ of a vertex 
$v\in V(G)$ is the set $\{u~|~(u,v)\in E(G)\}$. The {\it closed neighbourhood} $N_G[v]$ of a vertex $v\in V(G)$ is the set
$\{v\} \cup N_G(v)$.  The {\it degree} of $v\in V(G)$ is $|N_G(v)|$ and is denoted by $d_G(v)$. 
The subgraph induced by $D\subseteq V(G)$ is denoted by $G[D]$. 
For $X,Y \subseteq V(G)$ such that $X\cap Y = \emptyset$, $E_G(X,Y)$ denote the edges of $G$ with one endpoint in $X$ and the other in $Y$.
We will drop the subscripts in the above notation, whenever it is clear from the context.
For $s,t \in V(G)$, by an $st$-path in $G$ we mean a path from $s$ to $t$ in $G$. For $S,T \subseteq V(G)$,
by an $ST$-path we mean a path between a vertex of $S$ to a vertex of $T$.
For $i \in \mathbb{N}$, $[i]$ denotes the set $\{1, . . . , i\}$.

Let $G$ be a graph. 
A pair $(X, Y )$, where $X \cup Y = V(G)$, is called a \emph{separation} if $E(X \setminus Y, Y \setminus X) = \emptyset$. The order of $(X, Y )$ is $ |X \cap Y|$.
If there exists a separation $(X, Y )$ of order at most $k$ such that $|X \setminus Y | \geq q$ and $|Y \setminus X| \geq q$, then $G$ is {\em $(q, k)$-breakable} and the separation $(X,Y)$
 is called a {\em witnessing separation} for the $(q,k)$-breakability of $G$.
Otherwise, $G$ is {\em $(q, k)$-unbreakable}.

We refer to \cite{marekcygan} for the formal definition of Counting Monadic Second Order (CMSO) logic.
We will crucially use the following result of Lokshtanov et al.~\cite{Lokshtanov2018ReducingCM} that allows one to show that a CMSO-expressible graph problem is FPT by designing an FPT algorithm for the problem on $(q,k)$-unbreakable graphs, for any $k$ and a sufficiently large $q$ that depends only on $k$.

\prop*

Next, we state the Sunflower Lemma which is used in this paper. 
We first define the terminology used in the statement of the next lemma. 
Given a set system $(U,\mathcal{F})$, where $U$ is a set and $\mathcal{F}$ is a family containing distinct subsets of $U$,
a \emph{sunflower} with $k$ petals and a core $Y$ is a collection of sets $S_1, S_2, \dots, S_k \in \mathcal{F}$ 
such that $S_i \cap S_j = Y$ for all $i \neq j$.
The sets $S_i \setminus Y$ are called the \emph{petals} of this sunflower.
If the sets in $\mathcal{F}$ are distinct and $k \geq 2$, then none of the petals of the sunflower are empty. 
Note that a family of pairwise disjoint sets is a sunflower (with an empty core).

\begin{theorem}[Sunflower Lemma,~\cite{erdos1960intersection,marekcygan}]\label{thm:sunflower} Let $\mathcal{F}$ be a family of distinct sets over a universe $U$, such that each set in  $\mathcal{F}$ has cardinality exactly $d$. 
If $|\mathcal{F}| > d!(k - 1)^d$, then  $\mathcal{F}$ contains a sunflower with $k$ petals and such a sunflower can be computed in time polynomial in $|\mathcal{F}|$, $|U|$, and $k$.
\end{theorem}

\subsection{NP-hardness of \textsc{Maximum Minimal OCT}}\label{MMOCT is NP-hard}

\begin{lemma}
    {\sc Maximum Minimal OCT} is NP-hard.
\end{lemma}

\begin{proof} It was shown in \cite{HANAKA2019294} that {\sc Maximum Weight Minimal $st$-Separator} is NP-hard on bipartite graphs, even when all vertex weights are identical.
This implies that {\sc Maximum Minimal $st$-Separator} is NP-hard on bipartite graphs. 
We provide a polynomial-time reduction from the {\sc Maximum Minimal $st$-Separator} to the {\sc Maximum Minimal OCT}. 
Given an instance \(I = (G, s, t, k)\) of the {\sc Maximum Minimal $st$-Separator}, we construct an instance \(I' = (G', k' = k)\) of the {\sc Maximum Minimal OCT} as follows. We assume that $k>1$.
We consider two cases based on the bipartition of \(G\):

\noindent \textbf{Case 1:} If \(s\) and \(t\) are on the same side of the bipartition, we add an edge between \(s\) and \(t\), and subdivide it by adding two vertices \(u\) and \(v\), creating a new graph \(G'\).

\noindent \textbf{Case 2:} If \(s\) and \(t\) are on opposite sides of the bipartition, we add a subdivided edge between \(s\) and \(t\) with one new vertex \(u'\), resulting in \(G'\).

 In both cases, the newly added path between \(s\) and \(t\) in \(G'\) is denoted by \(P'\).

 We now prove that the instances \(I\) and \(I'\) are equivalent.

 In the forward direction, let \(Z\) be a minimal $st$-separator in \(G\) of size at least \(k\). We claim that \(Z\) is a minimal oct in \(G'\). Since \(G\) is bipartite, any odd cycle in \(G'\) must involve \(s\) and \(t\). Removing \(Z\) from \(G\) separates \(s\) and \(t\), meaning \(G'-Z\) contains only the newly added path \(P'\), ensuring no odd cycles in \(G'\). Thus, \(Z\) is an oct in \(G'\).

 Next, we show that \(Z\) is minimal. For every \(z \in Z\), the graph \(G - (Z \setminus \{z\})\) contains a path between \(s\) and \(t\). Lets call it $P$. In Case 1, this path is even-length, and in Case 2, it is odd-length. In both cases, the graph
\(G' - (Z \setminus \{z\})\) contains two vertex-disjoint paths $P$ nad $P'$ between \(s\) and \(t\) of different parities, forming an odd cycle. Therefore, \(Z\) is a minimal oct in \(G'\).

 In the backward direction, let \(Z'\) be a minimal oct in \(G'\) of size at least \(k\). Since \(k > 1\), \(Z'\) does not include the newly added vertices \(u, v\) (in Case 1) or \(u'\) (in Case 2). We claim that \(Z'\) is a minimal $st$-separator in \(G\).

 If \(Z'\) were not an $st$-separator in \(G\), then there would exist a path between \(s\) and \(t\) in \(G - Z'\), implying the presence of an odd cycle in \(G' - Z'\) involving the path \(P'\). 
 Since \(Z'\) is a minimal oct in \(G'\), for every $z\in Z'$, the graph \(G' - (Z'\setminus \{z\}) \) contains an odd cycle. Therefore, \(G' - (Z'\setminus \{z\}) \) must contain a path $P$ which is a vertex-disjoint $st$-path from $P'$. This implies that \(G - (Z' \setminus \{z\})\) contains an $st$-path, proving that \(Z'\) is a minimal $st$-separator in \(G\).
 Thus, \(I\) and \(I'\) are equivalent, completing the proof. 
\end{proof}

\section{\textsc{Maximum Minimal $st$-Separator} parameterized by the solution size}\label{st separator section} 

The goal of this section is to prove Theorem~\ref{Mainthm}.

\thmstsep*

Let $(G,k)$ be an instance of {\sc MaxMin $st$-separator}.
The goal is to reduce the task to designing an algorithm for $(q,k)$-unbreakable graphs.
For this we first show that the problem can be expressed in CMSO (in fact in MSO).
Since {\sc MaxMin $st$-Separator} is a maximization problem, 
the size of the solution can potentially be as large as $\mathcal{O}(n)$.
To formulate a CMSO sentence that is bounded by a function of $k$, 
we focus on a $k$-sized subset of the solution and encode the minimality of each of its vertices in a way that allows for its extension to a ``full-blown'' minimal solution. 

\begin{lemma}\label{minimal st separator CMSO}
   \textsc{Maximum Minimal $st$-Separator} is CMSO-definable with a formula of length $\mathcal{O}(k)$.
\end{lemma}
\begin{proof}
The instance $(G,k)$ is a yes instance of {\sc MaxMin $st$-Sep} if there
exists $Z \subseteq V(G)$ of size at least $k$ such that 
$G-Z$ has no $st$-path (equivalently $s$ and $t$ are in different connected components of $G-Z$), and
for each $v\in Z$, $G-(Z\setminus \{v\})$ has an $st$-path (that is $s$ and $t$ are in the same connected component of $G-(Z\setminus \{v\})$).

Alternately, suppose $Z \subseteq V(G)$ (of arbitrarily large size) 
such that $G-Z$ has no $st$-path,
and $Z$ contains $k$ distinct vertices 
$v_1, \ldots, v_k$ such that for each $i \in [k]$, 
$G-(Z\setminus \{v_i\})$ contains an $st$-path.
Then $Z$ may not be a minimal $st$-separator but it always contains a minimal $st$-separator of size at least $k$.
In fact, $(G,k)$ is a yes instance if and only if such a set $Z$ exists. 
These properties of $Z$ can be incorporated as a CMSO formula $\psi$ as follows, where $\mathbf{conn}(U)$ is a CMSO sentence that checks whether a vertex set $U$ induces a connected graph. The CMSO description of CMSO can be, for example, found in~\cite{marekcygan}.

\begin{align*}
   \psi  =&  \exists{Z \subseteq V(G)} \Bigg( 
   \exists v_1, v_2, \dots, v_k \in Z \left ( \bigwedge\limits_{1\leq i<j\leq k} v_{i} \neq v_{j} \right ) \\ 
&  \land \neg\exists {U \subseteq V(G)\setminus Z} 
\Bigg (  ( s \in U  ) \land  ( t \in U  ) \land \mathbf{conn}(U) \Bigg ) \\ 
 &\bigwedge\limits_{i=1}^{k}  \exists U \in V(G) \setminus (Z \setminus \{v_{i}\} )  \Bigg (    (s \in U )   \land   (t\in U)    \land \mathbf{conn}(U)  \Bigg )  \Bigg ) 
\end{align*}

It is clear that the size of the above formula $\psi$ depends linearly on $k$. 
\end{proof}

From Lemma~\ref{minimal st separator CMSO} and Proposition~\ref{cmso}, to prove Theorem~\ref{Mainthm},
it is enough to prove Theorem~\ref{thm:st-unbreak}.

\thmstsepunbreak*

We prove Theorem~\ref{thm:st-unbreak} in Section~\ref{sec:st-unbreak}.
As mentioned earlier, given a vertex set $V'$, it may not always be possible to extend it to a minimal $st$-separator. Below, we give a definition for a \textit{certificate} for the $st$-separator minimality of a set $V'\subseteq V(G)$, which guarantees the existence of a minimal $st$-separator that contains (extends) the set $V'$.

\begin{definition}\label{def:certificate of minimality}
    Let $G$ be a graph, $s,t \in V(G)$ and $V' \subseteq V(G)$. 
    We say that two sets of vertices $S$ and $T$ serve as a \textit{certificate} for the $st$-separator minimality for $V'$ if the following conditions hold: $s\in S$ and $t\in T$, $S \cap T = \emptyset$, $G[S]$ and $G[T]$ are connected subgraphs, $E_G(S,T) = \emptyset$, and for every $v \in V'$, the subgraph $G[S \cup T \cup \{v\}]$ is connected. 
    Note that $V' \cap (S \cup T) = \emptyset$.
\end{definition}

\begin{lemma}\label{greedy alg minimal st separator}
    Let $G$ be a graph, and let $s, t \in V(G)$. If there exists a \textit{certificate} for the $st$-separator minimality of $V' \subseteq V(G)$, then there exists a minimal $st$-separator in $G$ that includes all the vertices of $V'$.
\end{lemma}
\begin{proof} 
Let $S$ and $T$ serve as a certificate for the $st$-separator minimality of $V'$. 
We will construct a set $V' \subseteq Z \subseteq V(G) \setminus (S \cup T)$ which is a minimal $st$-separator in $G$.
The set $Z$ is constructed iteratively. Initialize $Z:=\emptyset$ and $G':= G[S \cup T]$.
Fix an arbitrary ordering of the vertices in 
 $V(G) \setminus (S \cup T)$.

For each vertex $v$ in the prescribed order: (i) if  $G[S \cup T \cup \{v\}]$ is connected, then update $Z := Z \cup \{v\}$, otherwise (ii) if $G[S \cup T \cup \{v\}]$ is not connected, then update  $G' := G[V(G') \cup \{v\}]$, update $S$ to be the vertices reachable from old $S$ in $G'$, and update $T$ to be the vertices reachable from old $T$ in $G'$.

The process continues until all vertices in $V(G) \setminus V(G')$  have been processed. The final set $Z$ is then returned as a minimal $st$-separator of $G$. Note that if the initial sets $S$ and $T$ served as a \textit{certificate} for the minimality of the $st$-separator  $V'$, then  $V' \subseteq Z$, 
as the subgraph $G[S \cup T \cup \{v\}]$ will always be connected for each $v \in V'$  during every stage of the above process. 
\end{proof}

\subsection{\textsc{Maximum Minimal $st$-Separator} on $(q,k)$-unbreakable graphs}\label{sec:st-unbreak}

In this section, we prove Theorem~\ref{thm:st-unbreak}.
To prove Theorem~\ref{thm:st-unbreak} we design a branching algorithm that maintains a tuple $(G,S,T,k,q)$ where $G$ is a $(q,k)$-unbreakable graph, $S,T \subseteq V(G)$ and $q,k$ are positive integers. Additionally the sets $S,T$ satisfy the following properties.

\begin{enumerate}
    \item $s \in S$ and $t \in T$,
    \item $S \cap T=\emptyset$,
    \item both $G[S]$ and $G[T]$ are connected and
    \item $E(S,T) =\emptyset$. 
\end{enumerate}

An instance $(G,S,T,k,q)$ satisfying the above properties is called a {\em valid} instance.
Given a valid instance,
we design a branching algorithm that outputs a minimal $ST$-separator of $G$,
which is disjoint from $S \cup T$, 
and has size at least $k$, if it exists.
The algorithm initializes $S:=\{s\}$ and $T:=\{t\}$.
Note that the above-mentioned properties of the sets $S,T$ ensure that at each stage of the algorithm, 
every minimal $ST$-separator is also a minimal $st$-separator.

The algorithm has one reduction rule (Reduction Rule~\ref{Red1}), four stopping criteria (Reduction Rules~\ref{Red4},~\ref{Red3},~\ref{Red2} and~\ref{Red5}) and one branching rule (Branching Rule~\ref{bred:one}).
The branching rule is applied when neither the reduction rule nor the three stopping criterion can be applied.
The overall idea is the following. Observe that $N(S) \cap N(T)$ is a part of any minimal $ST$-separator. Therefore, if $|N(S) \cap N(T)| \geq k$, then we can correctly report that $G$ has a minimal $ST$-separator of size at least $k$.
Reduction Rule~\ref{Red1} ensures that $N(S)$ (resp.~$N(T)$) is a {\em minimal} $ST$-separator.
Therefore when Reduction Rule~\ref{Red1} is no longer applicable, if $|N(S)| \geq k$ (resp.~$|N(T)| \geq k$), then we can correctly report that $G$ has a minimal $ST$-separator of size at least $k$. In fact, $|N(S)|$ (resp.~$|N(T)|$) is a minimal $ST$-separator of size at least $k$.
Otherwise, we have that both $|N(S)| < k$ and $|N(T)| <k$.
In this case, we use the branching rule.
Say, without loss of generality that $|S| \leq |T|$.
Since $N(S)$ is a minimal $ST$-separator, but its size is strictly less than $k$ and every minimal $ST$-separator contains $N(S) \cap N(T)$, 
there exists a vertex in $N(S) \setminus N(T)$ that does not belong to the solution (if there exists a solution of size at least $k$).
In this case we branch on the vertices of $N(S)\setminus N(T)$. 
If we guess that a vertex $v \in N(S)$ does not belong to the solution, since $v \in N(S)$, $v$ remains reachable from $S$ after removing the solution. 
In this case, we update $S:= S \cup \{v\}$. 
Therefore, each application of the branching rule increases the size of the smaller of the two sets $S$ or $T$.
When both $|S| \geq q$ and $|T| \geq q$, from the $(q,k)$-unbreakability of $G$, we know that {\em every} $ST$-separator of $G$ has size at least $k$. And hence there is a minimal $st$-separator of size at least $k$.

Below we formalize the above arguments.
Given $I = (G,S,T,k,q)$, we define its measure $\mu(I) = q - \min(|S|,|T|)$. 
We state the reduction rules and a branching rule below. We apply the reduction rules in order exhaustively before applying the branching rule.

\begin{lemma}\label{lem:neg-measure}
    Let $I=(G,S,T,k,q)$ where $G$ is $(q,k)$-unbreakable.
    If $\mu(I)\leq 0$, then every minimal $ST$-separator,
    which is disjoint from $S \cup T$, is of size at least $k$.
\end{lemma}
\begin{proof}
If $\mu(I) \leq 0$, then $q \leq \min(|S|,|T|)$.
For the sake of contradiction, let us assume that there exists a minimal $ST$-separator $Z$ in $G$, which is disjoint from $S \cup T$, and has size strictly less than $k$. 
Note that $G\setminus Z$ contains two connected components of size at least $q$ each: one containing $S$, say $C_S$, and the other containing $T$. 
Consider the separation $(C_S \cup Z, V(G) \setminus C_S)$ of $G$. This is a witnessing separation that $G$ is $(q,k)$-breakable, which is a contradiction.
\end{proof}

The safeness of Reduction Rule~\ref{Red4} is immediate from Lemma~\ref{lem:neg-measure}.

\begin{red}\label{Red4}
    If $\mu(I) \leq 0$, then report a yes instance.    
\end{red}

\begin{red}\label{Red3}\rm
   If $|N(S) \cap N(T)| \geq k$, then report a yes instance.
\end{red}

\begin{lemma}\label{lem:safe-red3}
    Reduction Rule \ref{Red3} is safe.
\end{lemma}
\begin{proof} The above reduction rule is safe because 
the sets $S$ and $T$ serve as a \textit{certificate} for the $st$-separator minimality of the vertex set $N(S)\cap N(T)$. 
From Lemma \ref{greedy alg minimal st separator}, there exists a minimal $st$-separator in $G$ that contains $N(S)\cap N(T)$. Since $|N(S) \cap N(T)| \geq k$,
we conclude that $G$ contains a minimal $st$-separator of size at least $k$. 
\end{proof}

\begin{lemma}\label{lem:one}
    If there exists $v\in N(S)$ (resp.~$v \in N(T)$),
    such that every path from $v$ to any vertex of $T$ (resp.~$S$)
    intersects $N(S) \setminus \{v\}$ (resp.~$N(T) \setminus \{v\}$),
    or there is no path from $v$ to any vertex of $T$ (resp.~$S$),
    then there is no minimal $ST$-separator which is disjoint from $S \cup T$ and that contains $v$.
\end{lemma}
\begin{proof}
    When $v$ has no path to any vertex of $T$, then such a vertex cannot lie on any $ST$-path and hence, is not a part of any minimal $ST$-separator.

    Suppose now that $v \in N(S)$ and every path from $v$ to any vertex of $T$ intersects
    $N(S) \setminus \{v\}$. The other case when $v \in N(T)$ is symmetric.
     For the sake of contradiction, say there exists a minimal $ST$-separator $Z$ such that $v \in Z$ and $Z \cap (S \cup T) = \emptyset$. 
     This implies that in the graph $G - Z$, there is no path from any vertex in $S$ to any vertex in $T$, 
     but in the graph $G - (Z \setminus \{v\})$, such a path, say $P$, exists.
     Let $P$ be a path from $s'$ to $t'$ in $G-(Z\setminus \{v\})$, where $s' \in S$ and $t' \in T$.

    Since $v$ cannot reach any vertex of $T$ (in particular $t'$) in $G$,
    without traversing another vertex, say $u$, in $N(S)$, 
    consider the $u$ to $t'$ subpath of $P$ (which does not contain $v$).
    Since $u \in N(S)$, let $s'' \in N(u) \cap S$. 
    Since $Z \cap (S \cup T) = \emptyset$,
    there exists a path $P'$ from $s''$ to $t'$ 
    in $G - (Z \setminus \{v\})$ that does not contain $v$ (take the edge $(s'',u)$,
    followed by the $u$ to $t'$ subpath of $P$). This contradicts that $Z$ is an $ST$-separator.
\end{proof}

The following reduction rule ensures that both $N(S)$ and $N(T)$ are minimal $ST$-separators.

\begin{red}\label{Red1}\rm
    If there exists $v\in N(S)$ (resp.~$v \in N(T)$),
    such that every path from $v$ to any vertex of $T$ (resp.~$S$)
    intersects $N(S)$ (resp.~$N(T)$),
    or there is no path from $v$ to any vertex of $T$ (resp.~$S$),
    then update $S:=S \cup \{v\}$ (resp.~$T:=T \cup \{v\}$).
\end{red}

\begin{lemma}
    Reduction Rule \ref{Red1} is safe.
\end{lemma}
\begin{proof}
From Lemma~\ref{lem:one},
no minimal $ST$-separator, that is disjoint from $S\cup T$, contains $v$.
Since $v \in N(S)$ (resp.~$v \in N(T)$),
for any minimal $ST$-separator $Z$,
$v$ is reachable from $S$ in $G-Z$, since $Z \cap S = \emptyset$ (resp.~$Z \cap T = \emptyset$). Thus, $Z$ is also a minimal separator between $S \cup \{v\}$ and $T$.
\end{proof}

\begin{lemma}\label{lem:minimalST}
    When Reduction Rule~\ref{Red1} is no longer applicable,
    $N(S)$ (resp.~$N(T)$) is a minimal $ST$-separator.
\end{lemma}
\begin{proof}
    First note that $N(S)$ (resp.~$N(T)$) is an $ST$-separator in $G$ which is disjoint from $S \cup T$, since $N(S) \cap T = \emptyset$ because $E(S,T) = \emptyset$. 

    For the sake of contradiction, say $N(S)$ is not a minimal $ST$-separator in $G$. In particular, there exists $v \in N(S)$ such that $N(S) \setminus \{v\}$ is also an $ST$-separator. 
    Since Reduction Rule~\ref{Red1} is no longer applicable,
    there exists a path from $v$ to a vertex of $T$, say $t'$, 
    which has no other vertex of $N(S)$. 
    Such a path together with an edge from $v$ to a vertex of $S$, gives an $ST$-path, which intersects $N(S)$ only at $v$.
\end{proof}

From Lemma~\ref{lem:minimalST}, the safeness of Reduction Rule~\ref{Red2} is immediate.

\begin{red}\label{Red2}\rm
   If $|N(S)|\geq k$ or $|N(T)|\geq k$ then report a yes instance.
\end{red}

\begin{red}\label{Red5}
    If $N(S) \setminus N(T) = \emptyset$ or $N(T) \setminus N(S) = \emptyset$, then report a no instance.
\end{red}

\begin{lemma}
    Reduction Rule~\ref{Red5} is safe.
\end{lemma}
\begin{proof}
    Suppose $N(S) \setminus N(T) = \emptyset$. The other case is symmetric.
    Then any $ST$-path uses only the vertices of $S \cup T \cup (N(S) \cap N(T))$. In this case there is a unique $ST$-separator in $G$ which is disjoint from $S \cup T$. This separator is $N(S) \cap N(T)$. Since Reduction Rule~\ref{Red3} is no longer applicable, $|N(S) \cap N(T)| <k$. Thus $G$ has no $ST$-separator of size at least $k$.
\end{proof}

\begin{bred}\label{bred:one}
    If $|S| \leq |T|$ (resp.~$|T| < |S|$) and $N(S) \setminus N(T) \neq \emptyset$ (resp.~$N(T) \setminus N(S) \neq \emptyset$), then for each vertex $x \in N(S) \setminus N(T)$ (resp.~$x \in N(T) \setminus N(S)$), we recursively solve the instance $(G, S \cup \{x\}, T, k,q)$ (resp.~$(G, S, T \cup \{x\}, k,q)$).
\end{bred}

First observe that the new instances created in this branching rule are all valid, that is, the sets $S \cup \{x\}$, $T$ (respectively $S$, $T \cup \{x\}$) satisfy the desired properties: the two sets $S \cup \{x\}$ and $T$ (resp.~$S$ and $T \cup \{x\}$) are disjoint, $G[S \cup \{x\}]$ (resp.~$G[T \cup \{x\}]$) is connected and  $E_G(S \cup \{x\},T) = \emptyset$ (resp.~$E_G(S ,T \cup \{x\})$) because $x \in N(S) \setminus N(T)$.

\begin{lemma}
    Branching Rule \ref{bred:one} is exhaustive, that is, $G$ has a minimal $ST$-separator of size at least $k$ which is disjoint from $S \cup T$ if and only if there exists $x \in N(S) \setminus N(T)$ (resp.~$x \in N(T) \setminus N(S)$) such that $G$ has a minimal separator of size at least $k$ between $S \cup \{x\}$ and $T$ (resp.~$S$ and $T \cup \{x\}$), which is disjoint from $S \cup T \cup \{x\}$.
\end{lemma} 

\begin{proof} 
Assume that $|S| \leq |T|$. The other case is symmetric.
Since Reduction Rule~\ref{Red2} is no longer applicable, $|N(S)| < k$.
Since Reduction Rule~\ref{Red1} is no longer applicable and because of Lemma~\ref{lem:minimalST},
$N(S)$ is a minimal $ST$-separator.
Also because $N(S) \cap N(T)$ is contained in every minimal $ST$-separator in $G$ (from the proof of Lemma~\ref{lem:safe-red3}),
 for any minimal $ST$-separator in $G$ of size at least $k$, say $Z$, which is disjoint from $S \cup T$, there exists  a vertex $x \in N(S) \setminus N(T)$ such that $x \not \in Z$. Since $x \in N(S)$, $x$ remains reachable from $S$ in $G-Z$. In particular, $Z$ is also a minimal separator between $S \cup \{x\}$ and $T$.

In the other direction say the instance $(G,S \cup \{x\},T,k,q)$ reports a minimal separator, say $Z$, between $S \cup \{x\}$ and $T$ of size at least $k$. Since $S \cup \{x\}$ and $T$ satisfy the desired properties of a valid instance, $Z$ is also a minimal $ST$-separator in $G$. 
\end{proof}

\begin{proof}[Proof of Theorem~\ref{thm:st-unbreak}]
Initialize an instance $I = (G, S, T, k, q)$ where $S = \{s\}$ and $T=\{t\}$. Note that $\mu(I) = q -1$.
Apply Reduction Rules~\ref{Red4}-\ref{Red5} exhaustively in-order. 
Without loss of generality, say $|S| \leq |T|$, the other case is analogous.
Since none of the reduction rules are applicable, $1 \leq |N(S) \setminus N(T)| < k$.
Now apply Branching Rule~\ref{bred:one}.
After every application of the Branching Rule~\ref{bred:one},
 $\mu(I')$, where $I'$ is the new instance, strictly decreases if $|S| \neq |T|$.
If $|S| = |T|$, then after every two applications of the Branching Rule~\ref{bred:one}, the measure $\mu$ decreases by $1$.
From Reduction Rule~\ref{Red4}, if $\mu$ of an instance is at most $0$, then we stop and report a yes instance.
The correctness of this algorithm follows from the safeness of the Reduction Rules~\ref{Red4}-\ref{Red5} and Branching Rule~\ref{bred:one}.
We now argue about the running time.

Note that all reduction rules can be applied in polynomial time. Also all reduction rules, except Reduction Rule~\ref{Red1}, is applied only once throughout the algorithm. Reduction Rule~\ref{Red1} is applied at most $2q$ times (or until both $S$ and $T$ grow to a size of $q$ each). Thus only polynomial time is spent on all applications of all reduction rules.
The branching rule branches in at most $k-1$ instances and has depth bounded by $2q$. Therefore the overall running time is $(k-1)^{2q} \cdot n^{\mathcal{O}(1)}$.
\end{proof}

\section{{\sc Maximum Minimal  OCT} parameterized by solution size}\label{minimal OCT section}

 In this section, we prove Theorem~\ref{thm:oct}.

\thmoct*

Throughout this section, we call a set of vertices of $G$ whose deletion results in a bipartite graph, an {\em oct} of $G$.
To prove Theorem~\ref{thm:oct}, we first show that the problem is CMSO definable (Lemma~\ref{lem:oct-cmso}). 
Using Proposition~\ref{cmso}, one can reduce to solving  this problem on $(q,2k)$-unbreakable graphs. 
On $(q,2k)$-unbreakable graphs, 
we then list and prove two sufficient conditions (Lemmas~\ref{lot of small odd cycles} and~\ref{long induced odd cycle}) 
which always imply a yes instance 
(in fact the first one implies a yes instance even when the input graph is not $(q,2k)$-unbreakable). We also make an observation about irrelevant vertices that can be deleted without changing the solution of the instance (Observation~\ref{delete unnecessary vertex}).
Even though checking whether any one of these sufficient conditions hold or finding these irrelevant vertices, may not be efficient, 
nonetheless we design an FPT algorithm that 
correctly concludes that at least one of the sufficient conditions is met, or outputs an irrelevant vertex, whenever the number of vertices in the graph is strictly more than a number with is a function of $q$ and $k$ (Theorem~\ref{main vertex OCT}). 
In the case when an irrelevant vertex is outputted, deleting them reduces the size of the graph but the resulting graph may not be $(q,2k)$-unbreakable. In this case, we start from the beginning and solve the problem from scratch (on general graphs).
If none of the above hold, then 
 the number of vertices in the graph is bounded, and we can solve the problem using brute-force.

\begin{lemma}\label{lem:oct-cmso}
   \textsc{Maximum Minimal OCT} is CMSO-definable by a formula of length $\mathcal{O}(k)$.
\end{lemma}
\begin{proof}
Let $(G,k)$ be an instance of \textsc{Maximum Minimal OCT}. 
Then $(G,k)$ is a yes instance if there exists $Z \subseteq V(G)$ of size at least $k$ such that $G-Z$ is bipartite and 
for each $v \in Z$, $G- (Z \setminus \{v\})$ is not bipartite.

Alternately, let $Z \subseteq V(G)$ such that $Z$ contains $k$ distinct vertices $v_1, \ldots, v_k$, such that $G-Z$ is bipartite and
for each $i \in [k]$, $G- (Z \setminus \{v\})$ is not bipartite.
Observe that if such a set $Z$ exists, it may not be a minimal oct of $G$, but it definitely contains a minimal oct of size at least $k$.
In fact, $(G,k)$ is a yes instance if and only if such a set $Z$ exists.
We phrase this description of $Z$ as the CMSO formula $\psi$ as defined below.

\begin{align*}
  \varphi \equiv \exists Z \subseteq V(G) \, &\Bigg( 
\exists v_1, v_2, \dots, v_k \in Z \ \bigg( \bigwedge\limits_{1\leq i<j\leq k} v_{i} \neq v_{j} \bigg) \\
  &\wedge \textbf{bipartite}(V(G) \setminus Z) \\
  &\wedge \left( \bigwedge_{i=1}^{k} \neg \textbf{bipartite}(V(G) \setminus (Z \setminus \{v_i\}) \right) \Bigg)  
\end{align*}

 where \textbf{bipartite}($W$) is a CMSO sentence given below, which checks whether the graph induced by the vertices in $W$ is bipartite.

\begin{align*}
\textbf{bipartite}(W) \equiv &\exists X \subseteq W , \exists Y \subseteq W  \\
&\Bigg ( \left( X \cap Y = \emptyset \right ) \wedge \left (X \cup Y = W \right ) \\
&\wedge \forall u, v \in W \, \left (E(u, v) \implies (u \in X \iff v \in Y) \right ) \Bigg ).
\end{align*} 

 It is clear that the size of the above formula $\varphi$ depends linearly on $k$. 
\end{proof}

\begin{proof}[Proof of Observation~\ref{delete unnecessary vertex}]
For the forward direction,
let $Z$ be a minimal oct of $G$ of size at least $k$. 
Then
$G - Z$ is bipartite.
Additionally, for any $z \in Z$, 
the graph $G - (Z \setminus \{z\})$ contains an odd cycle through $z$, implying the existence of an induced odd cycle $C_z$ containing $z$.
As the cycle is an induced odd cycle, we know that the vertex $v$ does not lie on this cycle.
We will show that $Z$ is also a minimal oct in $G-v$.
Clearly, $G-(Z \cup \{v\})$ is a bipartite graph.
Therefore, $Z$ is an oct of $G-v$.
For any vertex $z\in Z$, 
there is an induced odd cycle $C_{z}$ in $G-(Z \cup \{v\})$.
Therefore, $Z$ is also a minimal odd cycle traversal of $G-v$.

 For the other direction, let $Z$ be a minimal oct of $G - v$ with $|Z| \geq k$. 
Then, $G - (Z \cup \{v\})$ is bipartite.
For each $z \in Z$, the graph $G - ((Z \setminus \{z\}) \cup \{v\})$ contains an induced odd cycle $C_z$.
Since $v$ is not part of any induced odd cycle, $G - Z$ is bipartite, and $Z$ remains a minimal odd cycle transversal in $G$. 
\end{proof}

Because of Lemma~\ref{lem:oct-cmso} we can invoke Proposition~\ref{cmso}.
We invoke Proposition~\ref{cmso} with $c=2k$. Let $q$ be the $s$ from this proposition that corresponds to this choice of $c$.
We conclude that to prove Theorem~\ref{thm:oct} it is enough to prove Theorem~\ref{thm:oct-unbreak}.

\thmoctunbreak*

We prove Theorem~\ref{thm:oct-unbreak} in Section~\ref{sec:oct-unbreak}.
Next we define a certificate for minimality of oct for a vertex set $V'$. The existence of such certificates guarantees the existence of a minimal oct which contains $V'$.

\begin{definition}
    Given a graph $G$ and a set of vertices $V' \subseteq V(G)$, 
    we say that an induced subgraph $G'$ of $G$ is a {\it certificate} for the oct-minimality of $V'$,
    if $G'$ is bipartite and
    for every $v\in V'$, 
    $G[V(G') \cup \{v\}]$ contains an odd cycle.
\end{definition}

\begin{lemma}\label{greedy alg}
    Given a graph $G$ and a set of vertices $V'$, if there is a {\it certificate} for the oct-minimality of $V'$ then 
    there exists 
    a minimal oct of $G$ that contains all the vertices in $V'$.
\end{lemma}
\begin{proof} Let $G'$ be a certificate for minimality of $V'$.
First observe that $V' \cap V(G') =\emptyset$ because $G'$ is bipartite, but $G[V(G') \cup \{v\}]$, for any $v \in V'$, contains an odd cycle. 
We now construct a minimal oct of $G$, say $Z$,
iteratively as follows. 
Initialize $Z = \emptyset$. 
Fix an arbitrary ordering of the vertices in $V(G) \setminus V(G')$ (note that $V' \subseteq V(G) \setminus V(G')$). 
Traverse the vertices of $V(G) \setminus V(G')$ in this order.
For any vertex $v \in V(G) \setminus V(G')$ in the chosen order:
\begin{itemize}
    \item if $G[V(G') \cup \{v\}]$ is bipartite, update $G' := G[V(G') \cup \{v\}]$, that is add $v$ to $G'$, otherwise
    \item $G[V(G') \cup \{v\}]$ contains an odd cycle, in which case add $v$ to the set $Z$.
\end{itemize}

\noindent When all the vertices in $V(G) \setminus V(G')$ have been processed as stated above, then the set $Z$ is a minimal oct of $G$. Moreover, one can observe that if we start with $G'$ which a certificate for minimality of $V'$ then $V'\subseteq Z$. 
\end{proof}

\subsection{ \textsc{Maximum Minimal OCT} on $(q,2k)$-unbreakable graphs}\label{sec:oct-unbreak}

The goal of this section is to prove Theorem~\ref{thm:oct-unbreak}. 
Let $(G,k)$ be an instance of {\sc MaxMin OCT}.
A vertex $v \in V(G)$ is called {\em irrelevant} if $(G,k)$ is equivalent to $(G-v,k)$.
To prove Theorem~\ref{thm:oct-unbreak} it is enough to prove Theorem~\ref{main vertex OCT}.

\begin{theorem}\label{main vertex OCT}
    For positive integers $q,k \geq 1$,
    given as input a graph $G$ which is $(q,2k)$-unbreakable on at least $(2q+2)^{2}(2q+2)! (k-1)^{2q+2} +1$ vertices, 
    there exists an algorithm that runs in time $ (qk)^{\mathcal{O}(q)} \cdot n^{\mathcal{O}(1)}$, 
    and 
    either returns a minimal oct of $G$ of size at least $k$ or,
    outputs an irrelevant vertex $v$.
\end{theorem}

To see the proof of Theorem~\ref{thm:oct-unbreak} assuming Theorem~\ref{main vertex OCT},
observe that if the number of vertices is at most $(2q+2)^{2}(2q+2)! (k-1)^{2q+2}$,
then the problem can be solved using brute-force.
Otherwise, the algorithm of Theorem~\ref{main vertex OCT} either reports a yes instance, or finds an irrelevant vertex $v$, in which case, delete $v$ from the graph and solve the problem on $G-v$ (which is not necessarily $(q,2k)$-unbreakable).
The rest of the section is devoted to the proof of Theorem~\ref{main vertex OCT}.

\noindent{\bf Irrelevant vertices.}
Recall Observation \ref{delete unnecessary vertex} from Section~\ref{sec:intro}.

\obsirrelevant*

Note that we cannot explicitly design a (polynomial-time) reduction rule based on the above observation,  because determining whether there is an induced odd cycle containing a given vertex is NP-complete (see Lemma~\ref{existence of odd cycle}). 

\vspace{9pt}

\noindent{\bf Sufficient condition 1 [Long induced odd cycle in $G$].}

\begin{lemma}\label{long induced odd cycle}
    For any positive integers $q,k$,
    if $G$ is $(q,2k)$-unbreakable and 
    there exists an induced odd cycle in $G$ of length at least $2q+2$, 
    then $G$ has a minimal oct of size at least $k$.
\end{lemma}

\begin{proof}
Let $C$ be an induced odd cycle of length at least $2q+2$ in $G$. 
Let $x,y \in V(C)$ be arbitrarily chosen vertices such that $C \setminus \{x,y\}$ contains exactly two paths $S$ and $T$ each of length at least $q$ each. 
Moreover since $C$ is an odd cycle one of these two paths is odd and the other is even.
Without loss of generality, we can assume that $S$ is a path of even length and $T$ is a path of odd length. 
Let $Z_x:= \{x\}$ and $Z_y:=\{y\}$.

Below we define a procedure that iteratively grows the sets in $(S,T,Z_x,Z_y)$ while maintaining the following invaraints.

\begin{itemize}
    \item The sets $S,T$ and $Z_x \cup Z_y$ are pairwise disjoint. 
    \item $G[S]$ is connected, bipartite and $|S| \geq q$.
    \item $G[T]$ is connected, bipartite and $|T| \geq q$.
    \item $G[S \cup T \cup \{y\}]$ is a certificate for the oct-minimality for $Z_x$ (and in particular, $G[S \cup T \cup \{y\}]$ is bipartite).
    \item $G[S \cup T \cup \{x\}]$ is a certificate for the oct-minimality for $Z_y$ (and in particular, $G[S \cup T \cup \{x\}]$ is bipartite).
    \item $N(x) \cap S \neq \emptyset$, $N(x) \cap T \neq \emptyset$, $N(y) \cap S \neq \emptyset$ and $N(y) \cap T \neq \emptyset$.
\end{itemize}

Observe that the starting sets $(S,T,Z_x,Z_y)$ defined earlier satisfy these invariant. Note that in invariants 4 and 5, the sets $G[S \cup T \cup \{y\}]$ and $G[S \cup T \cup \{x\}]$, which serve as a certificate of minimality for $Z_{x}$ and $Z_{y}$ respectively, rely on the fact that $C$ is an induced odd cycle.
Since the iterative procedure grows these sets, the last property always hold. Also $G[S \cup T \cup \{x\} \cup \{y\}]$ contains the odd cycle $C$.
The idea is to grow these sets until $Z_x$ or $Z_y$ has size at least $k$. If this happens then we can use Lemma~\ref{greedy alg}, to conclude that $G$ has a minimal oct of size at least $k$. 

Since $G$ is $(q,2k)$-unbreakable and
$|S|,|T| \geq q$, we have that every $ST$-separator in $G$ has size at least $2k+1$. 
In particular, $|N(S) \cup N(T)| \geq 2k+1$.
Thus, if we guarantee that $(N(S) \cup N(T)) \subseteq (Z_x \cup Z_y)$, 
then $|Z_x \cup Z_y| \geq 2k+1$. Hence either $|Z_x| \geq k$ or $|Z_y| \geq k$.

Towards this we grow the sets in $(S,T,Z_x,Z_y)$ as follows. Let us call the vertices in $S \cup T \cup Z_x \cup Z_y$ as marked.

\begin{claim}\label{oddcycle}
   Let $v\in N(S)\cap N(T)$ be an unmarked vertex. Either $G[S \cup T \cup \{y,v\}]$ or $G[S \cup T \cup \{x,v\}]$ contains an odd cycle. 
\end{claim}
\proof 
Since both $G[S]$ and $G[T]$ are connected bipartite graphs, there exists a unique bipartition, say $S = A_S \cup B_S$ and $T = A_T \cup B_T$ of $S$ and $T$ respectively.
Recall that $x$ has neighbours in both the sets $S$ and $T$. 
Since the graph $G[S \cup \{x\}]$ is bipartite, without loss of generality, 
assume that $(N(x) \cap S) \subseteq B_S$.
Since $G[T \cup \{x\}]$ is also bipartite, without loss of generality assume that $N(x) \cap T \subseteq B_T$.
Since the graph $G[S \cup T \cup \{x\}]$ is connected and bipartite, we know that there is an unique bipartition of  $G[S \cup T \cup \{x\}]$ which is 
 $G[S \cup T \cup \{x\}] = \left( (A_S \cup A_T \cup \{x\}) \cup (B_S \cup B_T) \right)$. \\
 
Now, let us focus on the connected graph $G[S \cup T \cup \{y\}]$. 
Because $G[S \cup \{y\}]$ is bipartite, without loss of generality, we can assume that $(N(y) \cap S) \subseteq B_S$. 
We now show that $(N(y) \cap T) \subseteq A_T$.
 For the sake of contradiction, assume that this is not the case. This implies that there are two possibilities. If $y$ has neighbours in both the sets $A_{T}$ and  $B_{T}$ then it leads to a contradiction as $G[T \cup \{y\}]$ is bipartite.
If the neighbours of $y$ are contained in the set $B_{T}$ then it leads to a contradiction as it would imply that  $G[S \cup T \cup \{x, y\}]$ is bipartite. 

 Therefore, we get an unique bipartition $G[S \cup T \cup \{y\}] = \left( (A_S \cup B_T \cup \{y\}) \cup (B_S \cup A_T)\right)$ of $G[S \cup T \cup \{y\}]$.
As the vertex $v$ has neighbours in both the sets $S$ and $T$, there are four possibilities: (i) if $v$ has a neighbour in $A_{S}$ and $A_{T}$ then $G[S \cup T \cup \{y,v\}$ contains an odd cycle, (ii) if $v$ has a neighbour in $A_{S}$ and $B_{T}$ then $G[S \cup T \cup \{x,v\}$ contains an odd cycle, (iii) if $v$ has a neighbour in $A_{T}$ and $B_{S}$ then $G[S \cup T \cup \{x,v\}$ contains an odd cycle, or (iv) if $v$ has a neighbour in $A_{T}$ and $B_{T}$ then $G[S \cup T \cup \{y,v\}$ contains an odd cycle.
This finishes the proof of the claim. \claimqed

Let $v$ be an arbitrarily chosen unmarked vertex.

\begin{description}
    \item[Case 1: $v \in N(S) \cap N(T)$.] 
    From Lemma~\ref{oddcycle}, either $G[S \cup T \cup \{x,v\}]$ or $G[S \cup T \cup \{y,v\}]$ or both contain an odd cycle. 
    If $G[S \cup T \cup \{x,v\}]$ contain an odd cycle, then update $Z_y:= Z_y \cup \{v\}$.
    If $G[S \cup T \cup \{y,v\}]$ contain an odd cycle, then update $Z_x:= Z_x \cup \{v\}$.
    If both are true then update $Z_x:=Z_x \cup \{v\}$ and $Z_y:=Z_y \cup \{v\}$. The sets $S,T$ remain the same. 
    Observe that in either case, the updated sets satisfy all the invariants.
    \item[Case 2: $v \in N(S) \setminus N(T)$.]
    In this case, 
    if $G[S \cup \{v\}]$ is bipartite,
    then update $S:= S \cup \{v\}$. The other sets $T,Z_x,Z_y$ remain the same.
    Note that the updated sets (in particular $S$) maintains the invariant.
    Otherwise $G[S \cup \{v\}]$ 
    contain an odd cycle. 
    In this case update $Z_x:= Z_x \cup \{v\}$ and  $Z_y:= Z_y \cup \{v\}$.
    \item[Case 3: $v \in N(T) \setminus N(S)$.] This case is symmetric to Case 2.
\end{description}

We repeat the above process until we have sets $(S, T, Z_x,Z_y)$ such that all vertices in $N(S) \cup N(T)$ are marked.
In particular, in this case $(N(S) \cup N(T)) \subseteq (Z_x \cup Z_y)$. As argued earlier, this implies either $|Z_x| \geq k$ or $|Z_y| \geq k$. In both cases, we report that $G$ has a minimal oct of size at least $k$ from Lemma~\ref{greedy alg}.
\end{proof}

\vspace{9pt}

\noindent{\bf Sufficient condition 2 [Large family of short induced odd cycles].}
\begin{lemma}\label{lot of small odd cycles}
     Let $(G,k)$ be an instance of \textsc{Maximum Minimal OCT}. Let $d$ be any positive integer. 
     If $\mathcal{F}$ is a family containing distinct induced odd cycles of $G$ of length at most $d$ and $|\mathcal{F}|> d(d!) (k-1)^{d}$
     then $G$ has a minimal oct of size at least $k$.
\end{lemma}
\begin{proof} From the Sunflower Lemma (Theorem~\ref{thm:sunflower}), we can conclude that there exist at least $k$ induced odd cycles $\{F_1, F_2, \dots, F_k\} \subseteq \mathcal{F}$ such that $V(F_i) \cap V(F_j) = Y$ for all $i, j \in [k]$. As the cycles in $\mathcal{F}$ are induced odd cycles, the graph induced by the set of vertices in $Y$ must be bipartite. Note that $Y$ could possibly be empty.
We claim that $G$ has a minimal odd cycle transversal of size at least $k$. Such a minimal odd cycle transversal can be obtained by the greedy algorithm described in the proof of Lemma \ref{greedy alg}. We start with the induced subgraph $G[Y]$. Clearly, any minimal odd cycle transversal obtained by this algorithm will contain at least one vertex from $V(F_i) \setminus Y$ for each $i \in [k]$. 
Since the sets $V(F_i) \setminus Y$ for each $i \in [k]$ are disjoint, a minimal odd cycle transversal must have a size of at least $k$. \end{proof}

\vspace{9pt}

\noindent{\bf The combination lemma.} 

\begin{lemma}\label{lem:sumup}
    Given a graph $G$, a vertex $x \in V(G)$ and positive integers $d,k$,
    there is an algorithm that runs in $(kd)^{\mathcal{O}(d)} \cdot n^{\mathcal{O}(1)}$ time,
    and correctly outputs one of the following:

    \begin{enumerate}
        \item an induced odd cycle containing $x$, 
        \item an induced odd cycle of length at least $d$,
        \item a family $\mathcal{F}$ of distinct induced odd cycles, each of length at most $d-1$,  
        such that $|\mathcal{F}| \geq d\cdot d! \cdot(k-1)^{d}$,
        \item a determination that there is no induced odd cycle containing $x$ in $G$.
    \end{enumerate}
\end{lemma}
\begin{proof}
Suppose $G$ contains an induced odd cycle containing $x$. Let $C_x$ be one such cycle.
We design an iterative algorithm that maintains a pair $(G',\mathcal{F})$, where $G'$ is an induced subgraph of $G$ and $\mathcal{F}$ is a family of induced odd cycles of length at most $d-1$ in $G$, with the following additional properties.
The graph $G'$ is guaranteed to contain the cycle $C_x$, if $C_x$ existed in the first place in $G$, and
every induced odd cycle in $G'$ is distinct from any cycle in the family $\mathcal{F}$.

Initialize $G' := G$ and $\mathcal{F} = \emptyset$.
In each iteration, 
the algorithm finds an arbitrary induced odd cycle of $G$, say $F$, in polynomial time, if it exists. 
The following cases can now arise.
\begin{enumerate}
    \item The algorithm fails to find the cycle $F$. 
    That is $G'$ has no induced odd cycle. In this case, report that $G$ has no induced odd cycle passing through $x$.
    
    This is correct because if $G$ had such an induced odd cycle passing through $x$, then $C_x$ exists in $G$ and by the invariants of the algorithm $C_x$ also exists in $G'$, which is a candidate for the cycle $F$.
    
    \item If $F$ contains $x$, then return $F$ as the induced odd cycle containing $x$.
    
    \item If the length of $F$ is at least $d$, then return $F$ as an induced odd cycle of length at least $d$ in $G$.

    \item Otherwise, $F$ exists but does not contain $x$ and has a length of at most $d-1$. 

    By the invariants of the pair $(G',\mathcal{F})$, $F$ is distinct from all the cycles in $\mathcal{F}$.
    Update $\mathcal{F}$ by adding $F$ to it. Then the updated $\mathcal{F}$ satisfies the required invariants.
    
    To update $G'$ proceed as follows.
    Guess the intersection of $V(F)$ with $V(C_x)$ (in $G'$). 
    Let this be $F'$. 
    The number of guesses is bounded by $2^{|F|} \leq 2^{d-1}$. 
    Since $F$ does not contain $x$, $F$ is not equal to $C_x$. Hence, $F \setminus F'$ is non-empty.
    Update $G':= G' - (F\setminus F')$. 
    Observe that the updated $G'$ contains $C_x$ if the old $G_x$ contained it.
    Also every induced odd cycle in the updated $G'$ is distinct from $F$ (that has been newly added to $\mathcal{F}$) because a non-empty subset of $V(F)$ (that is $V(F) \setminus V(F')$) has been deleted in the updated $G'$.
\end{enumerate}

If the first, second or third condition mentioned above does not apply in each of the first $d \cdot d! \cdot (k-1)^{d}$ iterations, then at the end of the $i$-th iteration where $i = d \cdot d! \cdot (k-1)^{d}$, $|\mathcal{F}| = d \cdot d! \cdot (k-1)^{d}$. In this case, we return $\mathcal{F}$ as a large family of induced odd cycles of length at most $d-1$ of $G$.

This finishes the description and correctness of the algorithm. For the running time observe that in each iteration, the algorithm runs in polynomial time to find $F$ and check the first three conditions. The fourth condition in each iteration makes $2^{d-1}$ guesses and the number of iterations is at most $d \cdot d! \cdot (k-1)^d$. Therefore the overall running time is $(kd)^{\mathcal{O}(d)} \cdot n^{\mathcal{O}(1)}$.
\end{proof}

\begin{proof}[Proof of Theorem \ref{main vertex OCT}]
The algorithm for Theorem~\ref{main vertex OCT} proceeds as follows. Recall that $G$ is a $(q,k)$-unbreakable graph.
For each $x \in V(G)$,
run the algorithm of Lemma~\ref{lem:sumup} on input $(G,x,2q + 2,k)$.

 If for some $x\in V(G)$, the algorithm returns an induced odd cycle of length at least $2q+2$ or a family $\mathcal{F}$ of distinct induced odd cycles of length at most $2  q+1$ such that $|\mathcal{F}| \geq (2q+2)(2q+2)! (k-1)^{2q+2}$ 
 then we return a yes instance due to Lemma~\ref{long induced odd cycle} or~\ref{lot of small odd cycles}, respectively. 
 If the algorithm returns that there is no induced odd cycle containing $x$ in $G$ then we can delete $x$ from $G$ and solve the problem on the reduced graph. 
 This is safe because of Observation \ref{delete unnecessary vertex}. 
 Finally, if the algorithm returns an induced odd cycle containing $x$ of length at least $2 q +2$, then again report that it is a yes instance because of Lemma~\ref{long induced odd cycle}.

 If none of the above conditions hold, then for each $x \in V(G)$, the algorithm of Lemma~\ref{lem:sumup} returns an induced odd cycle containing $x$, say $C_x$, of length at most $2 q +1$. Note that the cycles returned for each $x \in V(G)$ in this case may not be distinct.

The following claim shows that if the number of vertices in $G$ is large (and there is an induced odd cycle of small length for each vertex of the graph), 
then there exists a large family containing {\em distinct} induced odd cycles of small length in $G$, in which case we can report a yes instance by Lemma~\ref{lot of small odd cycles}. 

\begin{claim}
    If the number of vertices in $G$ is at least $ (2q+2)^{2}(2q+2)! (k-1)^{2q+2} +1$, then $G$ contains a minimal oct of size at least $k$.
\end{claim}
\proof We will construct a $\mathcal{F}$ of distinct induced odd cycles in $G$ of length at most $2q+1$. For every vertex $x\in V(G)$, 
we take a cycle $C_{x}$. 
Unless the cycle $C_{x}$ is already in $\mathcal{F}$, we will update $\mathcal{F}=\mathcal{F} \cup \{C_{x}\}$. As the length of cycles in $C_{x}$ returned by the above algorithm is bounded by $2q+1$, if the number of vertices in $G$ is more than $(2q+1)(2q+2)(2q+2)! (k-1)^{2q+2} $ then $|\mathcal{F}|\geq (2q+2)(2q+2)! (k-1)^{2q+2}$. Due to Lemma \ref{lot of small odd cycles}, we return a yes instance.  
\claimqed

This finishes the proof of Theorem \ref{main vertex OCT}. 

\end{proof}

\section{Conclusion}\label{sec:conclusion}
In this paper, we established that both {\sc Maximum Minimal $st$-Separator} and {\sc Maximum Minimal Odd Cycle Transversal (OCT)} are fixed-parameter tractable parameterized by the solution size.
Instead of using treewidth-based win-win approaches, we design FPT algorithms for highly unbreakable graphs for both these problems.
While we have demonstrated the FPT nature of these problems, the challenge of developing efficient FPT algorithms remains open. 
Additionally, the edge-deletion version of the {\sc Maximum Minimal OCT} can be shown to be FPT using similar techniques, but with much simpler ideas. But designing an faster/explicit FPT algorithm even for this version remains an interesting direction for future research. 
Finally, the parameterized complexity of the weighted version of  {\sc Maximum  Minimal $st$-Separator} and {\sc Maximum Minimal Weight OCT} remain open as very large weights cause problems while formulating the problem in CMSO, in order to reduce it to the unbreakable case.

\bibliographystyle{plainurl}
\bibliography{bibliography}

\end{document}